\documentclass[11pt]{article}
\usepackage{times,amsfonts,latexsym,amssymb,epsfig}
\usepackage{amsthm,amstext,mathptmx,amsmath}
\usepackage{eucal}
\usepackage{fullpage}

\begin{document}
\bibliographystyle{plain}

\newtheorem{theorem}{Theorem}[section]
\newtheorem{corollary}[theorem]{Corollary}
\newtheorem{note}[theorem]{Note}
\newtheorem{lemma}[theorem]{Lemma}
\newtheorem{claim}[theorem]{Claim}
\newtheorem{fact}[theorem]{Fact}
\newtheorem{definition}[theorem]{Definition}

\newcounter{fignum}
\newcommand{\figlabel}[1]
	{\\Figure \refstepcounter{fignum}\arabic{fignum}\label{#1}}

\def\supp{\mbox{Supp}}

\def\Lin#1{\ensuremath{\textrm{Lin}({#1})}}
\def\Search#1{\ensuremath{\textrm{Search}({#1})}}
\def\TS#1#2#3{\ensuremath{\textrm{TS}^{#1}({#2},{#3})}}
\def\DISJ#1#2{\ensuremath{\textrm{DISJ}^{#1}_{#2}}}
\def\Promise#1#2{\ensuremath{\textrm{Sum-Promise}_{{#1},{#2}}}}
\def\Psearch{\ensuremath{\Pi_{\textrm{search}}}}
\def\Pstep{\ensuremath{\Pi_{\textrm{step}}}}
\def\Ppromise{\ensuremath{\Pi_{\textrm{promise}}}}

\def\Pcc#1{\ensuremath{\mathsf{P}_{#1}^{cc}}}
\def\RPcc#1{\ensuremath{\mathsf{RP}_{#1}^{cc}}}
\def\NPcc#1{\ensuremath{\mathsf{NP}_{#1}^{cc}}}
\def\BPPcc#1{\ensuremath{\mathsf{BPP}_{#1}^{cc}}}
\def\corr{\ensuremath{\mathrm{corr}}}
\def\disc{\ensuremath{\mathrm{disc}}}
\def\deg{\ensuremath{\mathrm{deg}}}
\def\E{\ensuremath{\mathbb{E}}}
\def\Lift#1#2{\ensuremath{\mathrm{Lift}({#1},{#2})}}
\def\OR{\ensuremath{\mathrm{OR}}}

\newcommand{\ignore}[1]{\relax}

\title{Separating NOF communication complexity classes RP and NP}

\author{Matei David\\
	{\normalsize Computer Science Department}\\
	{\normalsize University of Toronto}\\
	{\tt matei at cs toronto edu}
\and Toniann Pitassi\thanks{\ Research supported by NSERC.} \\
	{\normalsize Computer Science Department}\\
	{\normalsize University of Toronto}\\
	{\tt toni at cs toronto edu}
}

\maketitle

\begin{abstract}
We provide a non-explicit separation
of the number-on-forehead communication complexity classes
RP and NP when the number of players is up to
$\delta \cdot \log{n}$ for any $\delta < 1$.
Recent lower bounds on Set-Disjointness~\cite{ls-setdisjointness,ca-setdisjointness}
provide an explicit separation between these classes
when the number of players is only up to $o(\log\log{n})$.
\end{abstract}


\section{Introduction}

In the number-on-forehead (NOF) model of communication complexity,
$k$ players are trying to evaluate a function $F$ defined on $k n$ bits.
The input of $F$ is partitioned into $k$ pieces of $n$ bits each,
call them $x_1, \dots, x_k$,
and $x_i$ is placed, metaphorically, on the forehead of player $i$.
Thus, each player sees $(k-1) n$ of the $k n$ input bits.
The players communicate by writing bits on a shared blackboard
in order to compute $F$.
This model was introduced by~\cite{cfl:multiparty}
and it has many applications, including
circuit lower bounds~\cite{hg90,nw91},
time/space tradeoffs for Turing Machines,
pseudo-random number generators for space-bounded Turing Machines~\cite{bns89},
and proof system lower bounds \cite{bps:setdisj}.

In this model, a protocol is said to be ``efficient'' if it has
complexity $(\log{n})^{O(1)}$.
Correspondingly, \Pcc{k}, \RPcc{k}, \BPPcc{k} and \NPcc{k} are the classes of functions
having efficient deterministic, one-sided-error randomized,
(two-sided-error) randomized and nondeterministic protocols, respectively.
The usual inclusions between these classes apply,
so $\Pcc{k} \subseteq \RPcc{k} \subseteq \NPcc{k}$
and $\RPcc{k} \subseteq \BPPcc{k}$.
One of the most fundamental questions in NOF communication complexity
is to provide separations between these classes.
In \cite{bdpw-nof}, Beame et al. show that $\RPcc{k} \neq \Pcc{k}$
for $k \le n^{O(1)}$ players.
Recently, \cite{ca-setdisjointness,ls-setdisjointness}
 show that $\NPcc{k} \not\subset \BPPcc{k}$
(and thus, that $\NPcc{k} \neq \RPcc{k}$)
for $k \le o(\log\log{n})$ players.
Our main result in this paper is the following.

\begin{theorem}[Main Theorem]
\label{rp-vs-np-thm}
  $\NPcc{k} \not\subset \BPPcc{k}$ (and thus, $\NPcc{k} \neq \RPcc{k}$)
  for all $\delta < 1$ and all
  $k \le \delta \cdot \log{n}$.
\end{theorem}

Until very recently, it was far from clear how to obtain communication
complexity lower bounds in the number-on-forehead model
for any function that could separate nondeterministic
from randomized complexity. 
The difficulty can be described as follows.
The only method currently known for obtaining multiparty NOF lower bounds
is the discrepancy method \cite{bns89,raz:bns,chung-tetali}. 
Lower bounds using discrepancy
are obtained by showing that the function in question has
small discrepancy with respect to some distribution.
Unfortunately, it is not hard to see that every function
with small nondeterministic complexity has high discrepancy
with respect to every distribution (see, for example, Lemma~3.1 in \cite{ca-setdisjointness}.)
Thus, the discrepancy method
seemed doomed to failure and new techniques seemed to be required.

However, in very recent work, these difficulties were overcome
to obtain a surprisingly elegant lower bound for the Set-Disjointness
function \cite{ca-setdisjointness,ls-setdisjointness}.
The idea behind their proofs as well as ours is as follows.

In a recent paper, Sherstov \cite{sherstov-quantum} 
(and implicitly also in Razborov \cite{razborov-quantum}) 
applied the discrepancy
method in a more general way for the 2-player model in order to overcome the above
difficulties. The {\it generalized} discrepancy method was 
adapted to the number-on-forehead
model in \cite{ca-setdisjointness,ls-setdisjointness} and 
can be described at a high level as follows. 
Start with some candidate function $F$, where $F$ has
small nondeterministic complexity, and we want to prove that $F$ has
high randomized communication complexity.
Now come up with a function $G$ and a distribution $\lambda$
such that: (1) $F$ and $G$ are highly correlated with respect
to $\lambda$; and (2) $G$ has small discrepancy with respect to $\lambda$.
It is not hard to see that if such a $G$ can be found, then
since $G$ has small discrepancy, it requires large
randomized complexity, and moreover since $F$ and $G$ are
very correlated, this in turn implies lower bounds on the randomized
complexity of $F$ as well.

Thus, to use the generalized discrepancy method,
the problem is to come up with the functions $F$ and $G$.
To accomplish this, we will use another wonderful idea due to
Sherstov \cite{sherstov-stoc}, and substantially generalized to
apply to the number-on-forehead setting by Chattopadhyay \cite{chatto-focs}.
We consider special functions of the form $F^{\phi}$.
This will be a function on $(k+1)n$ bits, computed by $k+1$ players.
Player 0 receives an $n$-bit vector $x$.
Player $i$, for $1 \leq i \leq k$ gets an $n$-bit vector
$y_i$.
The function $\phi$ takes as input $y_1, \ldots,
y_k$ and outputs an $n$-bit string $z$,
where $z$ has exactly $m$ 1's. We will view
$\phi$ a selecting $m$ bits/indices of Player 0's input, $x$.
The function $F^{\phi}$ will be the $\OR$ function applied to the
$m$ bits of $x$ as specified by 
$\phi(y_1, \ldots y_k)$.
(In earlier terminology, the $k+1$ players will apply the
$\OR$ function to Player 0's {\it unmasked} input.)

Note that regardless of what function $\phi$ is chosen,
$F^{\phi}$ will have a small nondeterministic protocol.
Player 0 simply guesses an index $j$ that is one of the indices
chosen by $\phi$, and then any of the other players
can easily verify whether or not $x_j$ is 1 in that position.
When $\phi$ is the bitwise AND function, then
$F^{\phi}$ is the Set-Disjointness function.
We will show that for almost all $\phi$, the randomized
communication complexity of $F^{\phi}$ is large
as long as $k$ is at most a constant times $\log n$.
Because we will be working with a random $\phi$, as
a bonus, our argument is substantially simpler that
the previous bounds obtained for Set-Disjointness.

\section{Definitions and Notation}

\subsection{Communication Complexity}

In the number-on-forehead (NOF) multiparty communication complexity game 
\cite{cfl:multiparty} there are
$k$ players that are trying to collaborate to
compute a function $F: X_1\times \ldots\times X_k \rightarrow \{0,1\}$ where
each $X_i=\{0,1\}^n$.
The $kn$ input bits are partitioned into $k$ sets, each of size $n$.
For $(x_1,\ldots,x_k)\in \{0,1\}^{kn}$,
and for each $i$, player $i$ knows the values of all of the inputs
except for $x_i$ (which conceptually is thought of as being placed on player
$i$'s forehead).

The players exchange bits according to an agreed-upon
protocol, by writing them on a public blackboard. 
A \emph{protocol} specifies, for every possible blackboard
contents, whether or not the communication is over, the output if over
and the next player to speak if not. A protocol also specifies what
each player writes as a function of the blackboard contents and of the
inputs seen by that player. The \emph{cost} of a protocol is the maximum
number of bits written on the blackboard.

In a \emph{deterministic protocol}, the blackboard is initially empty.
A \emph{randomized protocol} of cost $c$ is simply a probability
distribution over deterministic protocols of cost $c$, which can be
viewed as a protocol in which the players have access
to a \emph{shared} random string.
A \emph{non-deterministic protocol} is one where an initial guess string
appears on the blackboard at the beginning of the protocol,
and the players are trying to verify that the output of the function is 1
in the usual sense: there exists a guess string where the output of the protocol is 1
if and only if the output of the function is 1.

The \emph{deterministic communication complexity of $F$}, written
$D_k(F)$, is the minimum cost of a deterministic protocol for $F$ that
always outputs the correct answer.
For $0 \leq \epsilon < 1/2$,
let $R_{k,\epsilon}(F)$ denote the
minimum cost of a randomized protocol for $F$ which, for every input,
makes an error with probability at most $\epsilon$ (over the choice of
the deterministic protocols).
The \emph{(two-sided-error) randomized communication complexity of $F$} is
$R_k(F) = R_{k,1/3}(F)$.
Let $R_{k,\epsilon}^1(F)$ denote the minimum cost
of a randomized protocol for $F$
which is correct on all 0-inputs,
and for every 1-input, it makes an error with probability at most $\epsilon$.
The \emph{one-sided-error randomized communication complexity of $F$} is
$R_k^1(F) = R_{k,1/3}^1(F)$.
The \emph{non-deterministic communication complexity of $F$},
written $N_k(F)$, is the minimum cost of a non-deterministic protocol for $F$.
We usually drop the subscript $k$ when the number of players is clear
from the context.

Since any function $F_n$ on $kn$ bits can be computed using only $n$ bits of
communication,
following~\cite{bfs86}, for sequences of functions $F=(F_n)_{n\in\mathbb{N}}$,
protocols are considered ``efficient'' or ``polynomial''
if only polylogarithmically many bits are exchanged.
Accordingly, let \Pcc{k}, \RPcc{k}, \BPPcc{k} and \NPcc{k}
denote the classes of function families $F$
for which $D_k(F_n), R_k^1(F_n), R_k(F_n)$ and $N_k(F_n)$
are $(\log{n})^{O(1)}$, respectively.

Even though the standard communication complexity definitions above
are given for functions with range $\{0,1\}$,
we find it more convenient to work with the range $\{-1,1\}$.
We transform the former into the latter by mapping
$0 \rightarrow 1$ (representing \emph{false})
and $1 \rightarrow -1$ (representing \emph{true}).
Thus, for example, when the range of $F$ is $\{-1,1\}$,
in a non-deterministic protocol the players are trying to verify
that the output of $F$ is -1.

The most important method to prove lower bounds for randomized
communication complexity uses the concept of discrepancy.
An \emph{$i$-cylinder} $\Gamma_i$ in $X_1\times \ldots \times X_k$
is a set such that for all $x_1\in X_1,\ldots,x_k\in X_k,x'_i\in X_i$ we have
$(x_1,\ldots,x_i,\ldots, x_k)\in \Gamma_i$ if and only if
$(x_1,\ldots,x'_i,\ldots,x_k)\in \Gamma_i$.
A \emph{cylinder intersection}
is a set of the form $\bigcap_{i=1}^k \Gamma_i$ where each $\Gamma_i$ is an $i$-cylinder
in $X_1\times \cdots \times X_k$.
For a set $S$, let $1_S$ be its characteristic function,
which is 1 if the input is in $S$ and 0 otherwise.
Let $\lambda$ be a distribution on the inputs of $F$.
The \emph{discrepancy of $F$ on $\Gamma$ under $\lambda$}
is $\disc_{k,\lambda}^{\Gamma}(F) = \left|
\E_{\overline{x} \sim \lambda}[F(\overline{x}) 1_\Gamma(\overline{x})] \right|$.
The \emph{discrepancy of $F$ under $\lambda$}
is $\disc_{k,\lambda}(F) = \max_{\Gamma}\disc_{k,\lambda}^\Gamma(F)$.
The \emph{standard discrepancy method} \cite{bns89}
connects the discrepancy of a function $F$
with its randomized communication complexity as follows:
for every distribution $\lambda$,
$R_{k,\epsilon}(F) \ge \log \left( \frac{1-2 \epsilon}{\disc_{k,\lambda}(F)} \right)$.

\subsection{Notation}

Throughout this paper, the functions whose communication complexity
we are analyzing are denoted by capital letters such as $F$.
As mentioned in the introduction, we 
will be restricting our attention to certain functions which are
constructed from a \emph{base} function, usually denoted by lower case $f$,
and a \emph{masking} function, usually denoted by $\phi$.
In general, $m$ denotes the size of the input to the base function $f$,
and the range of this function is $\{-1,1\}$.
A specific base function we will work with is the $\OR$ function,
which takes on the value -1 if and only if any of its input bits is 1.
The masking function $\phi$ takes as input $k$ strings of $n$ bits each,
usually denoted by $y_1, \dots, y_k$, and it's output is
an $m$-element subset of $[1,n]$. We always have $m \le n$.
Starting with a base function $f$ and a masking function $\phi$,
we construct a function $\Lift{f}{\phi}$ on $(k+1)n$ input bits as follows.
Given $n$-bit inputs $x, y_1, \dots, y_k$, $\phi$ is evaluated
on the latter $k$ inputs to select a set of $m$ bits in $x$
on which we apply $f$.
Formally, $\Lift{f}{\phi}(x,y_1, \dots, y_k) = f(x|\phi(y_1, \dots, y_k))$,
where for a set $S \subseteq [1,n]$,
$x|S$ denotes the substring of $x$ indexed by the elements in $S$.
We are interested in the communication complexity of $\Lift{f}{\phi}$
in the NOF model with $k+1$ players,
where player 0 gets $x$ and players 1 through $k$ get $y_1$ through $y_k$,
respectively.

\subsection{Correlation, Fourier Representation and Degree}

Let $f, g : \{0,1\}^m \rightarrow \mathbb{R}$.
Let $\mu$ be a distribution on the set $\{0,1\}^m$.
We define the \emph{correlation between $f$ and $g$ under $\mu$}
to be $\corr_{\mu}(f,g) = \E_{x \sim \mu}[f(x)g(x)]$.
Whenever we omit to mention a specific distribution
when computing the correlation, an expected value or a probability,
it is to be assumed that we are talking about the uniform distribution.

For $S \subseteq [1,m]$, let $\chi_S(x) = (-1)^{\sum_{i \in S}{x_i}}$
be the Fourier character of the set $S$.
Let $f : \{0,1\}^m \rightarrow \mathbb{R}$
and let $f_S = \corr(f, \chi_S)$.
Then $f(x) = \sum_{S \subseteq [1,m]}{f_S \chi_S(x)}$
is the Fourier representation of $f$.
The \emph{exact degree of $f$} is the size of the largest $S$
such that $f_S$ is non-zero.
The \emph{$\epsilon$-approximate degree of $f$},
denoted by $\deg_{\epsilon}(f)$
is the smallest $d$
for which there exists a function $g$ of exact degree $d$
such that $\max_x{|f(x)-g(x)|} \le \epsilon$.

\subsection{Set Families}

Let $\mathcal{S} = (S_1, \dots, S_z)$ be a multi-set
of $m$-element subsets of $[1,n]$.
Let the \emph{range} of $\mathcal{S}$, denoted by $\bigcup \mathcal{S}$,
be the set of indices from $[1,n]$
that appear in at least one set in $\mathcal{S}$.
Let the \emph{boundary} of $\mathcal{S}$, denoted by $\partial \mathcal{S}$,
be the set of indices from $[1,n]$
that appear in exactly one set in the collection $\mathcal{S}$.

\section{Statement of Results}

Our main technical result is the following.
\begin{theorem}
\label{non-rp-thm}
  Let $\delta < 1$ be a constant.
  Let $\epsilon = (1-\delta)/4$.
  Let $m = n^\epsilon$ and let $k \le \delta \cdot \log{n}$.
  There exists a function $\phi$ such that
  $R_{k+1}(\Lift{\OR}{\phi}) \ge n^{\Omega(1)}$.
\end{theorem}

\begin{proof}[Proof of Main Theorem~\ref{rp-vs-np-thm}
from Theorem~\ref{non-rp-thm}]
Consider the function $\phi$ whose existence is guaranteed
by Theorem~\ref{non-rp-thm}.
On the one hand, the Theorem implies that
$\Lift{\OR}{\phi} \notin \BPPcc{k+1}$.

On the other hand,
the following is a nondeterministic protocol
for $\Lift{\OR}{\phi}$: guess an index $i \in [1,n]$ using $\log{n}$ bits;
player 0 (the one holding $x$ on its forehead)
locally computes $\phi(y_1, \dots, y_k)$ and communicates a 1
if $i$ belongs to that set;
player 1 communicates a 1 if $x_i = 1$.
The cost of this protocol is $O(\log{n})$.
Easily, $\Lift{\OR}{\phi}(x,y_1, \dots, y_k) = -1$ iff
there exists a guess $i$ such that both players communicate a 1.
Thus, $\Lift{\OR}{\phi} \in \NPcc{k+1}$.
\end{proof}

\section{Proof of Main Result}

We obtain our lower bounds on the bounded-error
communication complexity of $\Lift{\OR}{\phi}$ using
an analysis that follows \cite{ca-setdisjointness}.
In their paper, Chattopadhyay and Ada analyze the
Set-Disjointness function,
and for that reason, their masking function $\phi$
must be the AND function.
In our case, intuitively, we allow $\phi$ to be a random function.
While our results no longer apply to Set-Disjointness,
we still obtain a separation between
\BPPcc{k} and \NPcc{k} because, no matter what masking function
is used, $\Lift{\OR}{\phi}$ always has a cheap nondeterministic protocol.

At a more technical level,
the results of \cite{ca-setdisjointness} become trivial when $k \ge \log\log{n}$
because of the relationship between $n$
(the size of the input to $F$) and $m$
(the number of bits the base function $\OR$ gets applied to.)
For their analysis to go through, they need $n = 2^{2^k}m^{O(1)}$.
In our case, $n = m^{O(1)}$ is sufficient, and this allows
our results to be non-trivial for $k \le \delta \log{n}$
for any $\delta < 1$.

\subsection{Overview of Proof}

As mentioned earlier, we will start with the base function
$f=\OR$ on $m$ input bits, $m < n$.
We lift the base function $f$ in order to obtain
the lifted function $F^{\phi} = \Lift{f}{\phi}$.
Recall that $F^{\phi}$ is a function on $(k+1)n$ inputs
with small nondeterministic complexity, and is obtained
by applying the base function (in this case the $\OR$ function)
to the unmasked bits of Player 0's input, $x$.
We want to prove that for a random $\phi$,
$F^{\phi}$ has high randomized communication complexity.

Paturi \cite{paturi-degree} proved
that no function that is a sum of
low-degree
Fourier characters can well-approximate the $\OR$ function. 
This implies that there exists a function $g$ (also on $m$ bits)
and a distribution $\mu$ over all $m$-bit inputs
such that the functions $g$ and $f=\OR$ are highly correlated over
$\mu$ and furthermore, $g$ is orthogonal to all small Fourier characters.
This is our Lemma~\ref{approx-orthog-lem},
and it was originally proved using duality
by Sherstov \cite{sherstov-quantum} in the 
context of 2-player lower bounds for quantum
communication complexity.

Now we lift the function $g$ in order to
get the function $G^{\phi} = \Lift{g}{\phi}$.
Define $\lambda$ to be a distribution over all $(k+1)n$-bit
inputs that is the natural extension of $\mu$.
Since $g$ and $f=\OR$ are highly correlated over $\mu$,
it is not hard to see (using the definitions and the fact that
$\lambda$ is the natural extension of $\mu$ 
to the lifted space) 
that the lifted versions, $F^{\phi}$
and $G^{\phi}$ are also highly correlated over $\lambda$.

By the generalized discrepancy method (Lemma~\ref{gen-disc-lem}),
in order to prove that the randomized complexity of
$F^{\phi}$ is high, it suffices
to prove that $G^{\phi}$ has small discrepancy.
This final step
is accomplished by Lemmas~\ref{small-q-lem}, \ref{gen-q-lem}, and \ref{prob-q-lem},
using two important properties of $g$ and $\phi$.
The crucial property of $g$ that we exploit is
that it is orthogonal to the space of all small Fourier characters.
This property will be used to prove Lemma~\ref{small-q-lem}.
Secondly, we want $\phi$ to behave like a random function
with respect to all sub-cubes. This second property is
exploited in order to prove Lemma~\ref{prob-q-lem}.
We now proceed with the formal proof.

\subsection{Proof of Main Theorem}

The following lemma is from \cite{sherstov-quantum}.
Intuitively it shows the following.
Let $f$ be a base function on $m$ bits,
and with the property that no function in the low-degree Fourier
subspace can
approximate $f$. (We
will be interested in $f=\OR$.)
The lemma states that this implies the existence of another function $g$
and a distribution $\mu$ such that $g$ is in the orthogonal
subspace of low-degree Fourier characters and $g$ well-approximates $f$.

\begin{lemma}[Orthogonality Lemma, Lemma~5.1 in \cite{ca-setdisjointness}]
  \label{approx-orthog-lem}
  If $f : \{0,1\}^m \rightarrow \{-1,1\}$ is a function
  with $\delta'$-approximate degree $d$,
  there exist a function $g : \{0,1\}^m \rightarrow \{-1,1\}$
  and a distribution $\mu$ on $\{0,1\}^m$ such that:
  \begin{itemize}
  \item[(i)] $\corr_\mu(g,f) \ge \delta'$; and
  \item[(ii)] for every $T \subseteq [1,m]$ with $|T| \le d$
  and every function $h : \{0,1\}^{|T|} \rightarrow \mathbb{R}$,
  $\E_{x \sim \mu}[g(x) \cdot h(x|T)] = 0$.
  \end{itemize}
\end{lemma}

The next lemma is the generalized discrepancy lemma
from \cite{ca-setdisjointness}.
It states that if two functions $F$ and $G$ are highly correlated,
and if $G$ has small discrepancy (and hence high communication complexity),
then the communication complexity of $F$ is also high.

\begin{lemma}[Generalized Discrepancy Lemma, Lemma~3.2 in \cite{ca-setdisjointness}]
  \label{gen-disc-lem}
  Let $Z = Z_1 \times \dots \times Z_k$.
  Let $F, G : Z \rightarrow \{-1,1\}$
  and let $\lambda$ be a distribution on $Z$ such that
  $\corr_\lambda(G,F) \geq \delta'$.
  Then, for every $\epsilon' < \delta'/2$,
  \[
  R_{k,\epsilon'}(F) \ge \log \left( \frac{\delta' - 2 \cdot \epsilon'}{\disc_{k,\lambda}(G)} \right).
  \]
\end{lemma}

The following lemma is standard and used in every discrepancy argument.
See \cite{bns89,raz:bns,chung-tetali} for details.

\begin{lemma}[The standard BNS argument]
  \label{bns-lem}
  Let $Z = X \times Y_1 \times \dots \times Y_k$ and let $F : Z \rightarrow \{-1,1\}$.
  Let $\Gamma \subseteq Z$ be a cylinder intersection.
  We write $\overline{y}$ for $(y_1, \dots, y_k)$.
  Then,
  \[
  \bigg(
  \E_{x, \overline{y}} \left[ F(x, \overline{y}) 1_\Gamma(x, \overline{y}) \right]
  \bigg)^{2^k}
  \leq
  \E_{\overline{y}^0,\overline{y}^1}
  \left[ \left|
  \E_x \left[ \prod_{u \in \{0,1\}^k} F(x,y_1^{u_1}, \dots, y_k^{u_k}) \right]
  \right| \right].
  \]
\end{lemma}

\medskip

Using the above lemmas, We will now prove Theorem~\ref{non-rp-thm}.
By \cite{paturi-degree}, $\deg_{5/6}(\OR) \ge c \sqrt{m}$
for some constant $c$.
By Lemma~\ref{approx-orthog-lem}, applied with $f = \OR$, there exist
a function $g$ and a distribution $\mu$ such that:
\begin{itemize}
\item[(i)] $\corr_\mu(g,\OR) \ge 5/6$; and
\item[(ii)] for every $T \subseteq [1,m]$ with $T \le c \sqrt{m}$
and every function $h : \{0,1\}^{|T|} \rightarrow \mathbb{R}$,
$\E_{x \sim \mu}[g(x) h(x|T)] = 0$.
\end{itemize}

For every masking function $\phi$,
let $F^\phi = \Lift{\OR}{\phi}$ and let $G^\phi = \Lift{g}{\phi}$.
As in \cite{ca-setdisjointness},
we define the distribution $\lambda$ on $\{0,1\}^{(k+1)n}$ as follows.
For $x \in \{0,1\}^n$ and
$\overline{y} = (y_1, \dots, y_k) \in \{0,1\}^{kn}$, let
\[
\lambda(x,\overline{y}) = \frac{\mu(x|\phi(\overline{y}))}{2^{(k+1)n-m}}.
\]
It can be easily verified that
$\corr_\lambda(G^\phi,F^\phi) = \corr_\mu(g,\OR) \ge 5/6$.
Thus, by Lemma~\ref{gen-disc-lem},
\[
R(F^\phi) \ge \log \left( \frac{5/6 - 2(1/3)}{\disc_\lambda(G^\phi)} \right)
= \log \left( \frac{1}{\disc_{\lambda}(G^\phi)} \right) - \Theta(1).
\]
Let $\Gamma$ be the cylinder intersection that witnesses
the discrepancy of $G^\phi$ under $\lambda$. Then,
\[
\disc_{\lambda}(G^\phi)
=
\disc_{\lambda}^\Gamma(G^\phi)
=
\left|
\E_{(x,\overline{y}) \sim \lambda} [ G^\phi(x,\overline{y}) 1_\Gamma(x,\overline{y}) ]
\right|
=
2^m
\left|
\E_{x,\overline{y}} [ \mu(x|\phi(\overline{y})) g(x|\phi(\overline{y})) 1_\Gamma(x,\overline{y}) ]
\right|
\]
where the last equality follows from the connection between $\lambda$
and the uniform distribution.
Finally, by Lemma~\ref{bns-lem}, we obtain
\[
\forall \phi,
\left( \disc_\lambda(G^\phi) \right)^{2^k}
\le
2^{m 2^k} \E_{\overline{y}^0,\overline{y}^1}
\left[
\left| \E_x\left[ \prod_{u \in \{0,1\}^k}
\mu(x|\phi(y_1^{u_1}, \dots, y_k^{u_k}))
g(x|\phi(y_1^{u_1}, \dots, y_k^{u_k})) \right] \right|
\right].
\]

It is at this point that we diverge from the analysis in \cite{ca-setdisjointness}.
Let $A = A(\overline{y}^0,\overline{y}^1)$ be the event
``$\exists i$ such that $y_i^0 = y_i^1$''.
Clearly, this event depends only on the choice of $\overline{y}^0$
and $\overline{y}^1$.
By a simple union bound,
$\Pr_{\overline{y}^0,\overline{y}^1}[A] \le k / 2^n = 2^{-n+log{k}}$.
Furthermore, $\Pr_{\overline{y}^0,\overline{y}^1}[\overline{A}] \le 1$,
and since $|\mu g| \le 1$,
$\mathbb{E}_{\overline{y}^0,\overline{y}^1}[\dots|\overline{A}] \le 1$.
Thus,
\[
\forall \phi, \left( \disc_\lambda(G^\phi) \right)^{2^k}
\le 2^{-n + m 2^k + \log{k}}
+ 2^{m 2^k} \mathbb{E}_{\overline{y}^0,\overline{y}^1}
\left[
\left| \mathbb{E}_x\left[ \prod_{u \in \{0,1\}^k}
\mu(x|\phi(y_1^{u_1}, \dots, y_k^{u_k}))
g(x|\phi(y_1^{u_1}, \dots, y_k^{u_k})) \right] \right|
\big| \overline{A}
\right].
\]

For the remaining part of the analysis,
we fix the choices of $\overline{y}^0$ and $\overline{y}^1$
in such a way that the event $A$ does not occur.
For $u \in \{0,1\}^k$, define
$S_u = S_u(\overline{y}^0,\overline{y}^1,\phi) = \phi(y_1^{u_1}, \dots, y_k^{u_k})$.
Let $\mathcal{S} = \mathcal{S}(\overline{y}^0,\overline{y}^1,\phi)$
be the multi-set $( S_u : u \in \{0,1\}^k )$.
Even though the sets $S_u$ and the multi-set $\mathcal{S}$ depend
on $\overline{y}^0, \overline{y}^1$ and $\phi$,
we will usually omit explicitly indicating this dependence
in our proofs in order to reduce the clutter.
We define \emph{the number of conflicts in $\mathcal{S}$}
to be $q(\mathcal{S}) = m 2^k - |\bigcup \mathcal{S}|$.
Intuitively, $|\bigcup \mathcal{S}|$ measures the range of $\mathcal{S}$,
while $m 2^k$ is the maximum possible value for this range.

We use the following three Lemmas to complete our proof.

\begin{lemma}
\label{small-q-lem}
  For every $\overline{y}^0, \overline{y}^1$ and $\phi$,
  if $\overline{A}(\overline{y}^0, \overline{y}^1)$
  and $q(\mathcal{S}(\overline{y}^0,\overline{y}^1,\phi)) < c \cdot \sqrt{m} \cdot 2^k / 2$,
  then
  \[
  \E_x\left[ \prod_{u \in \{0,1\}^k}{\mu(x|S_u(\overline{y}^0,\overline{y}^1,\phi))g(x|S_u(\overline{y}^0,\overline{y}^1,\phi))} \right] = 0.
  \]
\end{lemma}

\begin{lemma}
\label{gen-q-lem}
  For every $\overline{y}^0, \overline{y}^1$ and $\phi$,
  if $\overline{A}(\overline{y}^0, \overline{y}^1)$,
  \[
  \E_x\left[ \prod_{u \in \{0,1\}^k}{\mu(x|S_u(\overline{y}^0,\overline{y}^1,\phi))} \right]
  \le
  \frac{2^{q(\mathcal{S}(\overline{y}^0,\overline{y}^1,\phi))}}{2^{m \cdot 2^k}}.
  \]
\end{lemma}

\begin{lemma}
\label{prob-q-lem}
  For every $\overline{y}^0, \overline{y}^1$,
  if $\overline{A}(\overline{y}^0, \overline{y}^1)$,
  when $\phi$ is chosen at random,
  \[
  \Pr_\phi[ q(\mathcal{S}(\overline{y}^0,\overline{y}^1,\phi)) = q | \overline{A}(\overline{y}^0, \overline{y}^1) ]
  \le
  \left( \frac{m \cdot 2^k}{n} \right)^q.
  \]
\end{lemma}

Before proving these Lemmas, we complete the proof of our main Theorem.
Since the bound on $\disc_\lambda(G^\phi)$ holds for every $\phi$,
we can write
\[
\E_\phi \left[
\left( \disc_\lambda(G^\phi) \right)^{2^k}
\right]
\le
2^{-n + m 2^k + \log{k}}
+
2^{m 2^k}
\E_{\overline{y}^0,\overline{y}^1,\phi}
\left[ \left|\mathbb{E}_x\left[ \prod_{u \in \{0,1\}^k}{\mu(x|S_u)g(x|S_u)} \right]\right| \big| \overline{A} \right].
\]
Moreover,
\begin{eqnarray*}
& & \E_{\overline{y}^0,\overline{y}^1,\phi}
\left[ \left| \E_x \left[ \prod_{u \in \{0,1\}^k}{\mu(x|S_u)g(x|  S_u )} \right] \right| \big|  \overline{A}  \right]\\
~~~~~~&
\le
&
\sum_{q \ge 0}
\Pr_{\phi}[q(\mathcal{S}) = q| \overline{A} ]
\mathbb{E}_{\overline{y}^0,\overline{y}^1,\phi}
\left[ \left| \mathbb{E}_x\left[ \prod_{u \in \{0,1\}^k}{\mu(x|S_u)g(x|S_u)} \right] \right| \big|
\overline{A},q(\mathcal{S}) = q \right]\\
\text{(by Lemma~\ref{small-q-lem}) }
&
\le
&
\sum_{q \ge c \sqrt{m} 2^k/2}
\Pr_{\phi}[ q(\mathcal{S}) = q | \overline{A}]
\E_{\overline{y}^0,\overline{y}^1,\phi}
\left[ \left| \E_x\left[ \prod_{u \in \{0,1\}^k}{\mu(x|S_u)g(x|S_u)} \right] \right| \big|
\overline{A},q(\mathcal{S}) = q \right]\\
\text{(because $|g|=1$) }
&
\le
&
\sum_{q \ge c \sqrt{m} 2^k/2}
\Pr_{\phi}[q(\mathcal{S}) = q | \overline{A}]
\E_{\overline{y}^0,\overline{y}^1,\phi}
\left[ \left| \E_x\left[ \prod_{u \in \{0,1\}^k}{\mu(x|S_u)} \right] \right| \big|
\overline{A},q(\mathcal{S}) = q \right]\\
\text{(by Lemma~\ref{gen-q-lem}) }
&
\le
&
\sum_{q \ge c \sqrt{m} 2^k/2}
\Pr_{\phi}[q(\mathcal{S}) = q|\overline{A}]
\frac{2^q}{2^{m 2^k}}\\
\text{(by Lemma~\ref{prob-q-lem}) }
&
\le
&
\sum_{q \ge c \sqrt{m} 2^k/2}
\left( \frac{m 2^k}{n} \right)^q
\frac{2^q}{2^{m 2^k}}\\
&
=
&
\frac{1}{2^{m 2^k}}
\sum_{q \ge c \sqrt{m} 2^k/2}
\left( \frac{2 m 2^k}{n} \right)^q.
\end{eqnarray*}
We have chosen $\epsilon = (1-\delta)/4$,
so $1-\epsilon-\delta = 3\epsilon$.
Furthermore, $m = n^\epsilon$ and $k \le \delta \log{n}$,
so $m 2^k / n \le n^{-1+\epsilon+\delta} = n^{-3 \epsilon} < 1/4$
when $n$ is large enough.
Thus, $2 m 2^k / n < 1/2$.
Using $\sum_{q \ge q_0}{w^q} = w^{q_0} / (1-w) \le 2 w^{q_0}$
for $w < 1/2$, we obtain
\[
\E_{\overline{y}^0,\overline{y}^1,\phi}
\left[ \left|\E_x\left[ \prod_{u \in \{0,1\}^k}{\mu(x|S_u)g(x|S_u)} \right]\right| \big| \overline{A} \right]
\le
\frac{2^{1 - c\sqrt{m} 2^k/2}}{2^{m 2^k}}.
\]
Putting everything together,
\[
\E_\phi \left[
\left( \disc_\lambda(G^\phi) \right)^{2^k}
\right]
\le
2^{-n + m 2^k + \log{k}}
+
2^{m 2^k} 2^{-m 2^k} 2^{1 - c\sqrt{m}2^k/2}.
\]
For the exponent of the first term,
note that $\log{k} \le m 2^k$
and $n \ge 4 m 2^k$, so $-n + m 2^k + \log{k} \le -2 m 2^k$.
When $m$ is large enough,
$-2 m 2^k \le - c \sqrt{m} 2^k / 4$.
For the exponent of the second term,
note that $1 \le c \sqrt{m} 2^k / 4$ when $m$ is large enough,
so $1- c \sqrt{m} 2^k / 2 \le -c \sqrt{m} 2^k / 4$.
Thus, the sum of the two terms
is at most $2^{1-c \sqrt{m} 2^k /4}$.
When $m$ is large enough, $1 \le c \sqrt{m} 2^k / 8$,
so
\[
\E_\phi \left[
\left( \disc_\lambda(G^\phi) \right)^{2^k}
\right]
\le
2^{-c \sqrt{m} 2^k/8}.
\]
Therefore, there exists some $\phi$ such that
$\disc_\lambda(G^\phi) \le 2^{-c \sqrt{m} /8}$.
For this $\phi$,
\[
R(F^\phi) \ge \log \left( \frac{1}{\disc_\lambda(G^\phi)} \right) - \Theta(1)
\ge \Theta(1) \sqrt{m} = \Theta(1) n^\epsilon \ge n^{\Omega(1)}.
\]

\section{Proofs of Lemmas}

\begin{proof}[Proof of Lemma~\ref{small-q-lem}]
We write $S_u$ for $S_u(\overline{y}^0,\overline{y}^1,\phi)$
and $\mathcal{S}$ for $\mathcal{S}(\overline{y}^0,\overline{y}^1,\phi)$.
Assume $q(\mathcal{S}) < c \sqrt{m}2^k/2$.
Let $r(\mathcal{S}) = |\bigcup \mathcal{S}|$ be the size of the range of $\mathcal{S}$,
and let $b(\mathcal{S}) = |\partial \mathcal{S}|$ be the size
of the boundary of $\mathcal{S}$.
Note that $r(\mathcal{S}) - b(\mathcal{S}) \le q(\mathcal{S})$
because every $j \in \cup \mathcal{S} \setminus \partial \mathcal{S}$
occurs in at least 2 sets in $\mathcal{S}$,
thus contributes at least 1 to $q(\mathcal{S})$.
Furthermore, $r(\mathcal{S}) + q(\mathcal{S}) = m 2^k$.
Then, $b(\mathcal{S}) \ge r(\mathcal{S}) - q(\mathcal{S})
= m 2^k - 2 q(\mathcal{S}) > (m - c \sqrt{m})2^k$.
There are $2^k$ sets in the multi-set $\mathcal{S}$
so by the pigeonhole principle,
there exists $v$ such that $|S_v \cap \partial \mathcal{S}| > m - c \sqrt{m}$.
We can write
\[
\E_x \left[ \prod_{u \in \{0,1\}^k}{\mu(x|S_u) g(x|S_u)} \right]
=
\E_{x|S_v} \left[ \mu(x|S_v) g(x|S_v)
\E_{x|[1,n] \setminus S_v} \left[
\prod_{u \in \{0,1\}^k, u \neq v}{\mu(x|S_u) g(x|S_u)} \right] \right].
\]
Let $T = S_v \setminus \partial \mathcal{S}$.
So $|T| \le c \sqrt{m}$.
Let $h = \E_{x|[1,n] \setminus S_v} \left[
\prod_{u \neq v}{\mu(x|S_u) g(x|S_u)} \right]$.
Note that $h$ is a function that depends only on $x|T$.
Then, by the property (ii) of $g$ and $\mu$,
$\E_{x|S_v} [ \mu(x|S_v) g(x|S_v) h(x|T) ] = 0$.
\end{proof}

\begin{proof}[Proof of Lemma~\ref{gen-q-lem}]
We write $S_u$ for $S_u(\overline{y}^0,\overline{y}^1,\phi)$
and $\mathcal{S}$ for $\mathcal{S}(\overline{y}^0,\overline{y}^1,\phi)$.
We see that
\[
\E_x \left[ \prod_{u \in \{0,1\}^k}{\mu(x|S_u)} \right]
=
\E_{x|[1,n] \setminus \bigcup \mathcal{S}}
\left[ \E_{x|\bigcup \mathcal{S}} \left[ \prod_{u \in \{0,1\}^k}{\mu(x|S_u)} \right] \right]
=
\E_{x|\bigcup \mathcal{S}} \left[ \prod_{u \in \{0,1\}^k}{\mu(x|S_u)} \right].
\]
Every $u \in \{0,1\}^k$ can be interpreted as an integer in the range $[0,2^k-1]$.
With this in mind,
for $0 \le j \le 2^k-1$,
let $\mathcal{S}_j$ be the sub-multi-set of $\mathcal{S}$
consisting of the sets up to and including $S_j$,
$\mathcal{S}_j = ( S_0, \dots, S_j )$.
So, $\mathcal{S} = \mathcal{S}_{2^k-1}$.
Define $\mathcal{S}_{-1} = \emptyset$.
For $0 \le j \le 2^k-1$, let
$G_j = \E_{x|\bigcup \mathcal{S}_j}[ \prod_{i=0}^{j}{\mu(x|S_i)} ]$
and let
$H_j(x|S_j \setminus \partial \mathcal{S}_j)
= \E_{x|S_j \cap \partial \mathcal{S}_j} [ \mu(x|S_j) ]$.
Letting $G_{-1} = 1$, observe that, for $0 \le j \le 2^k-1$,
\[
G_j = \E_{x|\bigcup \mathcal{S}_{j-1}}
\left[ \left( \prod_{i=0}^{j-1}{\mu(x|S_i)} \right) H_j(x|S_j \setminus \partial \mathcal{S}_j) \right]
\le (\max(H_j)) \cdot G_{j-1}.
\]
To obtain a bound on $\max(H_j)$, consider
an arbitrary partition of $[1,m]$ into two sets $E, F$.
Let $\nu$ be a distribution on $[1,m]$,
and let $\rho(x|E) = \mathbb{E}_{x|F}[\nu(x)]$.
Then, $\rho(x|E) = \sum_{x|F}{2^{-|F|} \nu(x)}
= 2^{-|F|} \sum_{x|F}{\nu(x)} \le 2^{-|F|} = 2^{|E|-m}$,
simply using the fact that $\nu$ is a probability distribution.
Thus, $\max(H_j) \le 2^{|S_j \setminus \partial \mathcal{S}_j|-m}$.
Inductively,
\[
\mathbb{E}_x \left[ \prod_{i=0}^{2^k-1}{\mu(x|S_i)} \right] =
G_{2^k-1}
\le
\frac{2^{\sum_{j=0}^{2^k-1}|S_j \setminus \partial \mathcal{S}_j|}}{2^{m 2^k}}.
\]
Consider some index $z \in \bigcup \mathcal{S}$.
Suppose this index appears in $l$ sets
$S_{j_1}, \dots, S_{j_l}$ from $\mathcal{S}$,
with $j_1 < \dots < j_l$.
Then, this index contributes exactly $l-1$ to the expression
$\sum_{j=0}^{2^k-1}|S_j \setminus \partial \mathcal{S}_j|$,
once for every $j = j_2, \dots, j_l$
(for $j = j_1$, $z \in \partial \mathcal{S}_j$ because no set
before $S_j$ contains $z$.)
Since this holds for every index $z$,
we see that $\sum_{j=0}^{2^k-1}|S_j \setminus \partial \mathcal{S}_j| = q(\mathcal{S})$
and therefore
$\E_x[ \prod_{u \in \{0,1\}^k}{\mu(x|S_u)} ] \le 2^{q(\mathcal{S})-m 2^k}$.
\end{proof}

\begin{proof}[Proof of Lemma~\ref{prob-q-lem}]
Fix $\overline{y}^0, \overline{y}^1$ such that $\overline{A}$.
The multi-set $\mathcal{S}$ is constructed from the sets
$S_u = \phi(y_1^{u_1}, \dots, y_k^{u_k})$
for $u \in \{0,1\}^k$.
Since $A$ did not occur, the $2^k$ points where $\phi$
gets evaluated are distinct.
Furthermore, $\phi$ is chosen at random,
which is equivalent to choosing $2^k$ random $m$-element subsets of $[1,n]$.
We can overestimate the number of conflicts in $\mathcal{S}$ as follows.
Instead of choosing, for each subset, $m$ elements from $[1,n]$
\emph{without} replacement,
suppose we chose them \emph{with} replacement.
The number of conflicts we will obtain can only be larger than
in the original experiment or, equivalently,
the probability of obtaining a fixed number of conflicts
can only be greater in the second experiment.
The maximum range of $\mathcal{S}$ is $m 2^k$.
Every conflict in $\mathcal{S}$ arises when we select a previously selected
point from $[1,n]$.
Thus, the probability of each conflict is independently at most $m 2^k/n$.
The probability of obtaining $q$ conflicts is at most $(m 2^k/n)^q$.
\end{proof}

\bibliography{theory,toni,josh,nate}

\end{document}